\documentclass[twocolumn,%,superscriptaddress,
showpacs,preprintnumbers,amsmath,amssymb]{revtex4}
\usepackage{amssymb,amsmath,amsthm}
\usepackage{graphicx}
\newtheorem{Thm}{Theorem}
\newtheorem{Cor}{Corollary}

\newtheorem{Prop}{Proposition}
\theoremstyle{definition}

\newcommand{\bra}[1]{{\left\langle #1 \right|}}
\newcommand{\ket}[1]{{\left| #1 \right\rangle}}

\newcommand{\T}{\mbox{$\mathrm{tr}$}}

\begin{document}
%%%%%%%%%%%%%%%%%%%%%%%%%%%%%%%%%%%%%%%%%%%%%%%%%%%%%%%%%%%%%%%%%%%%%%%%%%
%                                                                        %
%                                 Title                                  %
%                                                                        %
%%%%%%%%%%%%%%%%%%%%%%%%%%%%%%%%%%%%%%%%%%%%%%%%%%%%%%%%%%%%%%%%%%%%%%%%%%
\title{Polygamy of multi-party $q$-expected quantum entanglement
}
\author{Jeong San Kim}
\email{freddie1@khu.ac.kr} \affiliation{
 Department of Applied Mathematics and Institute of Natural Sciences, Kyung Hee University, Yongin-si, Gyeonggi-do 446-701, Korea
}
\date{\today}

%%%%%%%%%%%%%%%%%%%%%%%%%%%%%%%%%%%%%%%%%%%%%%%%%%%%%%%%%%%%%%%%%%%%%%%%%%
%                                                                        %
%                              Abstract                                  %
%                                                                        %
%%%%%%%%%%%%%%%%%%%%%%%%%%%%%%%%%%%%%%%%%%%%%%%%%%%%%%%%%%%%%%%%%%%%%%%%%%
\begin{abstract}
We characterize the polygamy nature of quantum entanglement in multi-party systems in terms of $q$-expectation value for the full range of $q\geq 1$. By investigating some properties of generalized quantum correlations in terms of $q$-expectation value and Tsallis $q$-entropy, we establish a class of polygamy inequalities of multi-party quantum entanglement in arbitrary dimensions based on $q$-expected entanglement measure. As Tsallis $q$-entropy is reduced to von Neumann entropy, and $q$-expectation value becomes the ordinary expectation value when $q$ tends to $1$, our results encapsulate previous results of polygamy inequalities based on von Neumann entropy as special cases.
\end{abstract}

%\keywords{}
\maketitle

%%%%%%%%%%%%%%%%%%%%%%%%%%%%%%%%%%%%%%%%%%%%%%%%%%%%%%%%%%%%%%%%%%%%%%
%%%                                                                %%%
%%%                         Introduction                           %%%
%%%                                                                %%%
%%%%%%%%%%%%%%%%%%%%%%%%%%%%%%%%%%%%%%%%%%%%%%%%%%%%%%%%%%%%%%%%%%%%%%
\section{Introduction}
Quantum entanglement is a quintessential phenomenon of quantum mechanics showing the non-local nature of quantum states in multi-party quantum systems. As a quantum correlation among distinct parties, entanglement plays a central role in quantum information and computation theory with many applications~\cite{tele, qkd1, qkd2}. Thus it has been an important and even challenging task to have a proper way of quantifying entanglement to understand full characteristics of quantum entanglement in various quantum systems.

One property that makes quantum entanglement fundamentally different from other classical correlations is the restricted shareability and distribution of entanglement in multi-party quantum systems, namely, {\em monogamy} and {\em polygamy } relations of quantum entanglement~\cite{T04, KGS}.
Mathematically, the monogamy of quantum entanglement has been characterized as monogamy inequalities using various entanglement measures~\cite{ckw, ov, kds, KSRenyi, KimT, KSU}. These monogamy inequalities of entanglement show the mutually exclusive structures of entanglement shareability in multi-party quantum systems. The polygamy relation of quantum entanglement was also quantitatively characterized as polygamy inequalities in multi-party quantum systems~\cite{GBS, BGK, KimGP}.

By using the concept of {\em $q$-expectation value} for any nonnegative real parameter $q$, von Neumann entropy can be generalized into a one-parameter class of entropy functions, namely {\em Tsallis $q$-entropy}~\cite{tsallis, lv}. Because $q$-expectation value is theoretically consistent with the minimum cross-entropy principle, Tsallis $q$-entropy is considered to be more relevant to nonextensive statistical mechanics~\cite{LP, Abe}. Tsallis $q$-entropy can also be used to characterizes classical statistical correlations inherent in quantum states~\cite{rr, bpcp}.

In quantum entanglement theory, Tsallis $q$-entropy can be used to define a faithful entanglement measure because the property of {\em entanglement monotone} is guaranteed by the concavity of Tsallis $q$-entropy for $q>0$~\cite{vidal}. It is also known that some conditions on separability criteria of quantum states can be established based on Tsallis $q$-entropy~\cite{ar,tlb,rc}.

Here, we characterize the polygamy relation of quantum entanglement in multi-party systems in terms of $q$-expectation value for the full range of $q\geq 1$. We first recall the generalized definitions of various classical and quantum correlations in terms of Tsallis $q$-entropy and $q$-expectation value. By investigating some properties of generalized correlations in relation to classical-classical-quantum(ccq) states, we establish a class of polygamy inequalities of multi-party quantum entanglement in arbitrary dimensions in terms of the $q$-expected entanglement measure.

Due to the existence of equivalence among the monogamy and polygamy inequalities of $q$-expected quantum correlations~\cite{Kim19}, our results also guarantee several classes of monogamy and polygamy inequalities about $q$-expected entanglement and discord distributed in three-party quantum systems. As Tsallis $q$-entropy is reduced to von Neumann entropy, and $q$-expectation value becomes the ordinary expectation value when $q$ tends to $1$, our results encapsulate previous results of polygamy inequalities based on von Neumann entropy as special cases.

This paper is organized as follows. In Sec.~\ref{sec: q-exp}, we recall the definitions and properties of $q$-expected classical and quantum correlations in terms of Tsallis $q$-entropy and $q$-expectation value.
In Sec.~\ref{subsec: ccq}, we provide the definition of a ccq state in four-party quantum systems, and its properties related with the $q$-expected correlations.
In Sec.~\ref{subsec: upperlower}, we provide analytic upper and lower bounds of $q$-expected correlations in accordance with the ccq states for $q\geq2$. In Sec.~\ref{sec: poly}, we establish a class of polygamy inequalities of multi-party quantum entanglement in arbitrary dimensions based on $q$-expected entanglement measure $q\geq 1$.
Finally, we summarize our results in Sec.~\ref{Conclusion}.

\section{$q$-Expected Quantum Correlations}
\label{sec: q-exp}

Based on the {\em generalized logarithmic function}
\begin{eqnarray}
\ln _{q} x &=&  \frac {x^{1-q}-1} {1-q},
\label{qlog}
\end{eqnarray}
of the real parameter $q$ with $q\geq0$ and $~q \ne 1$, the {\em Tsallis $q$-entropy} of a
quantum state $\rho$ is defined as~\cite{tsallis,lv}
\begin{align}
S_{q}\left(\rho\right)=-\T \rho ^{q} \ln_{q} \rho.
\label{Qtsallis}
\end{align}
Tsallis $q$-entropy is concave for any nonnegative real parameter $q\geq0$, and it converges to von Neumann entropy as $q$ tends to 1,
\begin{equation}
\lim_{q\rightarrow 1}S_{q}\left(\rho\right)=-\T\rho \ln \rho=:S\left(\rho\right).
\end{equation}

For a quantum state $\rho$ with the spectrum $\{\lambda_i \}_{i}$~\cite{spec},
its Tsallis $q$-entropy can be written as
\begin{align}
S_{q}\left(\rho\right)=-\sum_{i}\lambda_{i}^q \ln _{q}\lambda_i,
\label{Ctsallis}
\end{align}
that is, the $q$-{\em expectation value} of the generalized logarithms. Thus, Tsallis $q$-entropy is a one-parameter generalization of von Neumann entropy based on the concept of $q$-expectation value for nonnegative real parameter $q$.

Using Tsallis $q$-entropy and the concept of $q$-expectation value, a class of bipartite entanglement measures has been introduced; for $q\geq0$ and a bipartite pure state $\ket{\psi}_{AB}$, its $q$-{\em expected entanglement}($q$-E) is defined as~\cite{Kim19}
\begin{equation}
{E}_{q}\left(\ket{\psi}_{AB} \right)=S_{q}(\rho_A),
\label{qEpure}
\end{equation}
where $\rho_A=\T _{B} \ket{\psi}_{AB}\bra{\psi}$ is the reduced density matrix of $\rho_{AB}$ on subsystem $A$.
For a bipartite mixed state $\rho_{AB}$, its $q$-E is defined as the minimum $q$-expectation value
\begin{equation}
E_{q}\left(\rho_{AB} \right)=\min \sum_i p^q_i E_{q}(\ket{\psi_i}_{AB}),
\label{qEmixed}
\end{equation}
over all possible pure state decompositions of $\rho_{AB}$,
\begin{equation}
\rho_{AB}=\sum_{i} p_i |\psi_i\rangle_{AB}\langle\psi_i|.
\label{decomp}
\end{equation}
As a dual quantity to $q$-E, $q$-{\em expected entanglement of assistance}($q$-EOA) was also defined as
\begin{equation}
E^a_{q}\left(\rho_{AB} \right)=\max \sum_i p_i^q E_{q}(\ket{\psi_i}_{AB}),
\label{qEOA}
\end{equation}
where the maximum is taken over all possible pure state
decompositions of $\rho_{AB}$.

Tsallis $q$-entropy converges to von Neumann entropy and the $q$-expectation value becomes ordinary expectation value when $q$ tends to 1, therefore we have
\begin{align}
\lim_{q\rightarrow1}E_{q}\left(\rho_{AB} \right)=&E_{\rm f}\left(\rho_{AB} \right)
\end{align}
and
\begin{align}
\lim_{q\rightarrow1}E^a_{q}\left(\rho_{AB}\right)=&E^a\left(\rho_{AB} \right),
\label{TsallistEOA}
\end{align}
where $E_{\rm f}\left(\rho_{AB}\right)$ is the {\em entanglement of formation}(EOF)~\cite{bdsw}, and
$E^a(\rho_{AB})$ is the entanglement of assistance(EOA) of $\rho_{AB}$~\cite{cohen}.

Let us consider more generalized quantum correlations based on $q$-expectation value and Tsallis $q$-entropy. For $q\geq0$ and a probability ensemble $\mathcal E = \{p_i, \rho_i\}$ of a quantum state $\rho$ (equivalently, a probability decomposition $\rho=\sum_{i}p_i\rho_i$ denoted by $\mathcal E$), its {\em Tsallis-$q$ difference} is defined as~\cite{Kim16T}
\begin{align}
\chi_q\left(\mathcal E\right)=S_q\left(\rho\right)-\sum_{i}p_{i}^q S_q\left(\rho_i\right).
\label{eq: q-diff}
\end{align}
Tsallis-$q$ difference is nonnegative for $q\geq 1$ due to the concavity of Tsallis $q$-entrop, and it converges to the Holevo quantity
\begin{align}
\lim_{q\rightarrow1}\chi_q\left(\mathcal E\right)=S\left(\rho\right)-\sum_{i}p_i S\left(\rho_i\right)=: \chi\left(\mathcal E\right),
\label{eq: holevo}
\end{align}
as $q$ tends to 1. 

For a bipartite quantum state $\rho_{AB}$, each measurement $\{M^x_B\}$ applied on subsystem $B$ induces a probability ensemble $\mathcal E = \{p_x, \rho_A^x\}$ of the reduced density matrix $\rho_A=\T_A\rho_{AB}$ in the way that
\begin{align}
p_x=\T[(I_A\otimes M_B^x)\rho_{AB}]
\end{align}
is the probability of the outcome $x$ and
\begin{align}
\rho^x_A=\T_B[(I_A\otimes {M_B^x})\rho_{AB}]/p_x
\end{align}
is the state of system $A$ when the outcome was $x$.
The {\em one-way classical $q$-correlation}
($q$-CC)~\cite{Kim19} of a bipartite state $\rho_{AB}$ is defined as the maximum Tsallis-$q$ difference
\begin{align}
{\mathcal J}_q^{\leftarrow}(\rho_{AB})&= \max_{\mathcal E} \chi_q\left(\mathcal E\right)
\label{qCC}
\end{align}
over all possible ensemble representations $\mathcal E$ of $\rho_A$ induced by measurements on subsystem $B$.

As a dual quantity to $q$-CC, the {\em one-way unlocalizable $q$-entanglement}($q$-UE)~\cite{Kim16T} is defined by taking the minimum  Tsallis-$q$ difference
\begin{align}
{\mathbf u}E_q^{\leftarrow}(\rho_{AB}) &= \min_{\mathcal E} \chi_q\left(\mathcal E\right),
\label{qUE}
\end{align}
over all probability ensembles $\mathcal E $ of $\rho_A$ induced by {\em rank-1 measurements} on subsystem $B$.
Due to the continuity of Tsallis-$q$ difference with respect to $q$, we have
\begin{align}
\lim_{q\rightarrow1}{\mathcal J}_q^{\leftarrow}(\rho_{AB})=&
{\mathcal J}^{\leftarrow}(\rho_{AB})
\end{align}
and
\begin{align}
\lim_{q\rightarrow1}{\mathbf u}E_q^{\leftarrow}(\rho_{AB})
=&{\mathbf u}E^{\leftarrow}(\rho_{AB}),
\label{contJq}
\end{align}
where ${\mathcal J}^{\leftarrow}(\rho_{AB})$ is the {\em one-way classical correlation}(CC)~\cite{KW} and ${\mathbf u}E^{\leftarrow}(\rho_{AB})$ is {\em one-way unlocalizable entanglement}(UE)~\cite{BGK} of the bipartite state $\rho_{AB}$.

The following proposition shows the trade-off relations between $q$-CC and $q$-E as well as $q$-UE and $q$-EOA distributed in three-party quantum systems.
\begin{Prop}~\cite{Kim19}
For $q \geq 1$ and a three-party pure state $\ket{\psi}_{ABC}$ with its reduced density matrices $\rho_{AB}=\T_C\ket{\psi}_{ABC}\bra{\psi}$,
$\rho_{AC}=\T_B\ket{\psi}_{ABC}\bra{\psi}$ and $\rho_{A}=\T_{BC}\ket{\psi}_{ABC}\bra{\psi}$, we have
\begin{align}
S_q(\rho_A)={\mathcal J}_q^{\leftarrow}(\rho_{AB})+E_q\left(\rho_{AC}\right)
\label{qCCEq}
\end{align}
and
\begin{align}
S_q(\rho_A)={\mathbf u}E_q^{\leftarrow}(\rho_{AB})+E^a_q\left(\rho_{AC}\right).
\label{qUEEqa}
\end{align}
\label{thm: qUEEqa}
\end{Prop}

The concept of $q$-expectation and Tsallis-$q$ entropy are also used to generalize {\em quantum discord}~\cite{discord}, a different kind of quantum correlation. For $q\geq0$ and a bipartite quantum state $\rho_{AB}$, its {\em Tsallis-$q$ mutual entropy} is defined as
\begin{align}
{\mathcal I}_q\left(\rho_{AB}\right)=S_q\left(\rho_A\right)+S_q\left(\rho_B\right)-
S_q\left(\rho_{AB}\right),
\label{eq: qmutul}
\end{align}
which generalizes the quantum mutual information
\begin{align}
{\mathcal I}\left(\rho_{AB}\right)=&S(\rho_A)+S(\rho_B)-S(\rho_{AB})
\label{mut}
\end{align}
in a way that
\begin{align}
\lim_{q\rightarrow1}{\mathcal I}_q\left(\rho_{AB}\right)={\mathcal I}\left(\rho_{AB}\right).
\label{qmut}
\end{align}

For a bipartite state $\rho_{AB}$, its {\em quantum $q$-discord}($q$-D)\cite{qdiscord} is defined by the difference between its Tsallis-$q$ mutual entropy and $q$-CC,
\begin{align}
\delta_q^{\leftarrow}(\rho_{AB})={\mathcal I}_q\left(\rho_{AB}\right)-{\mathcal J}_q^{\leftarrow}(\rho_{AB}).
\label{qdis}
\end{align}
We note that $q$-D is a generalization of the quantum discord $\delta^{\leftarrow}(\rho_{AB})$ as \begin{align}
\lim_{q\rightarrow1}\delta_q^{\leftarrow}(\rho_{AB})
={\mathcal I}\left(\rho_{AB}\right)-{\mathcal J}^{\leftarrow}(\rho_{AB})=:\delta^{\leftarrow}(\rho_{AB}).
\label{dis}
\end{align}
Moreover, the duality between $q$-CC and $q$-UE provides us with a dual definition to $q$-D,
\begin{align}
{\mathbf u}\delta_q^{\leftarrow}(\rho_{AB})={\mathcal I}_q\left(\rho_{AB}\right)-{\mathbf u}E_q^{\leftarrow}(\rho_{AB}).
\label{qudis}
\end{align}
Eq.~(\ref{qudis}) is referred to as the {\em one-way unlocalizable quantum $q$-discord}($q$-UD) of $\rho_{AB}$~\cite{Kim19},
which is a generalization of {\em one-way unlocalizable quantum discord}(UD)~\cite{XFL12}
\begin{align}
{\mathbf u}\delta^{\leftarrow}(\rho_{AB})={\mathcal I}\left(\rho_{AB}\right)-{\mathbf u}E^{\leftarrow}(\rho_{AB}).
\label{udis}
\end{align}

The following proposition provides a trade-off relation between quantum entanglement($q$-UE) and quantum discord($q$-UD) distributed in three-party quantum systems.
\begin{Prop}~\cite{Kim19}
For $q \geq 1$ and a three-party pure state $\ket{\psi}_{ABC}$ with its reduced density matrices $\rho_{AB}=\T_C\ket{\psi}_{ABC}\bra{\psi}$,
$\rho_{AC}=\T_B\ket{\psi}_{ABC}\bra{\psi}$ and $\rho_{A}=\T_{BC}\ket{\psi}_{ABC}\bra{\psi}$, we have
\begin{align}
S_q(\rho_A)={\mathbf u}\delta_q^{\leftarrow}(\rho_{BA})+{\mathbf u}E_q^{\leftarrow}(\rho_{CA}).
\label{qUEqD}
\end{align}
\label{thm: qUEqD}
\end{Prop}

\section{Some properties of $q$-expected quantum correlations}
\label{sec: Prop}
In this section, we first consider a class of four-party classical-classical-quantum(ccq) states and their $q$-expected correlations. By assuming the subadditivity of Tsallis-$q$ mutual entropy for this class of ccq states,
we provide an analytic upper bound for $q$-UE as well as a lower bound for $q$-UD.

\subsection{Classical-Classical-Quantum States}
\label{subsec: ccq}

For a two-qudit state $\rho_{AB}$, let us consider a spectral decomposition of
the reduced density matrix $\rho_B=\T_{A}\rho_{AB}$  such that
\begin{align}
\rho_B=\sum_{i=0}^{d-1}\lambda_{i}\ket{e_i}_B\bra{e_i}.
\label{specrhoB}
\end{align}
Based on the eigenvectors $\{ \ket{e_j }_{B}\}$ of $\rho_B$,
generalized $d$-dimensional Pauli operators can be defined as
\begin{align}
Z=\sum_{j=0}^{d-1}\omega_d^j\ket{e_j}\bra{e_j},~
X=\sum_{j=0}^{d-1} \omega_d^{-j}|\tilde
e_j \rangle \langle \tilde e_j |,
\label{paulis}
\end{align}
where $\omega_d = e^{\frac{2\pi i}{d}}$ is the $d$th-root of unity, and
$\{ |\tilde e_j \rangle _{B} \} $ is the $d$-dimensional {\em Fourier basis},
\begin{equation}
|\tilde e_j \rangle = \frac{1}{\sqrt{d}}\sum_{k=0}^{d-1}
\omega_d^{jk}\ket{e_k},~j=0,\ldots ,d-1,
\label{fourier}
\end{equation}
with respect to the eigenvectors $\{ \ket{e_j }_{B}\}$ of $\rho_B$.
By using the Pauli operators in Eq.~(\ref{paulis}), we define two quantum operations acting on
any $d$-dimensional quantum state $\sigma$ as
\begin{equation}
M_0(\sigma)=\frac{1}{d}\sum_{b=0}^{d-1}Z^b\sigma Z^{-b},~
M_{1}(\sigma)=\frac{1}{d}\sum_{a=0}^{d-1}X^a\sigma X^{-a}.
\label{channels2}
\end{equation}

For the two-qudit state $\rho_{AB}$ whose reduced density matrix is $\rho_B$ in Eq.~(\ref{specrhoB}),
the actions of the channels $M_{0}$ and $M_{1}$ applied on the subsystem $B$ are
\begin{align}
(I_A\otimes M_0)(\rho_{AB})&=\sum_{i=0}^{d-1} \sigma_A^i \otimes \lambda_i\ket{e_i}_B\bra{e_i}
\label{cact1}
\end{align}
and
\begin{align}
(I_A\otimes M_1)(\rho_{AB})&=\sum_{j=0}^{d-1} \tau_A^j \otimes
\frac{1}{d}|\tilde e_j \rangle_B \langle \tilde e_j|,
\label{cact2}
\end{align}
where $\lambda_i\sigma_A^i=\T_B [(I_A \otimes
\ket{e_i}_B\bra{e_i})\rho_{AB}]$ and $\tau_A^j/d=\T_B[(I_A
\otimes |\tilde e_j\rangle_B \langle\tilde e_j|)\rho_{AB}]$ for $i,~j \in \{0,\cdots, d-1\}$.
Thus the ensembles of subsystem $A$ induced by the action of the channels $M_0$ and $M_1$ from Eqs.~(\ref{cact1}) and
(\ref{cact2}) are
\begin{align}
\mathcal E_0=\{\lambda_i,\sigma_A^i\}_i,~
\mathcal E_1=\{\frac{1}{d},\tau_A^j\}_j,
\label{ensembles}
\end{align}
respectively. Equivalently, we can say that each of the rank-1 measurements $\{ \ket{e_i}_B\bra{e_i}\}$ and $\{ |\tilde e_j\rangle_B \langle\tilde e_j|\}$ on subsystem $B$ of $\rho_{AB}$ induces the ensembles $\mathcal E_0$ and $\mathcal E_1$ of subsystem $A$, respectively.

Now, let us consider a four-qudit ccq-state $\Omega_{XYAB}$
\begin{align}
  \Omega_{XYAB}=\frac 1{d^2}\sum_{x,y=0}^{d-1}&\ket{x}_X
  \bra{x}\otimes\ket{y}_Y\bra{y}\nonumber\\
  &\otimes(I_A\otimes X^x_BZ^y_B)\rho_{AB}(I_A\otimes
  Z^{-y}_BX^{-x}_B),
\label{XYAB}
\end{align}
with the reduced density matrices
\begin{align}
\Omega_{XAB}=\frac
1{d}\sum_{x=0}^{d-1}&\ket{x}_X\bra{x}\otimes X^x_B
\left(\sum_{i=0}^{d-1} \sigma_A^i \otimes \lambda_i\ket{e_i}_B\bra{e_i}\right)X_B^{-x}
\label{XAB}
\end{align}
and
\begin{align}
\Omega_{YAB}=&\frac
1{d}\sum_{y=0}^{d-1}\ket{y}_Y\bra{y}\otimes
Z_B^y\left(\sum_{j=0}^{d-1} \tau_A^j \otimes
\frac{1}{d}|\tilde e_j \rangle_B \langle \tilde
e_j|\right)Z_B^{-y}. \label{YAB}
\end{align}
It is straightforward to verify that the Tsallis-$q$ mutual entropies of $\Omega_{XYAB}$,
$\Omega_{XAB}$ and $\Omega_{YAB}$ in Eqs.~(\ref{XYAB}), (\ref{XAB}) and (\ref{YAB}) are~\cite{Kim16T}
\begin{align}
{\mathcal I}_q\left(\Omega_{XY:AB}\right)=&\frac{d^{1-q}-1}{1-q}+d^{1-q}S_q\left(\rho_A\right)
%\nonumber\\
-d^{2(1-q)}S_q\left(\rho_{AB}\right),
\label{IqXYAB2}
\end{align}
\begin{align}
{\mathcal I}_q\left(\Omega_{X:AB}\right)
=&\frac{d^{1-q}-1}{1-q}-d^{1-q}S_q\left(\rho_B\right)+d^{1-q}\chi_q(\mathcal E_0)
\label{IqXAB2}
\end{align}
and
\begin{align}
{\mathcal I}_q\left(\Omega_{Y:AB}\right)=(1-d^{1-q})\frac{d^{1-q}-1}{1-q}+d^{1-q}\chi_q(\mathcal E_1).
\label{IqYAB2}
\end{align}

\subsection{Upper and Lower Bounds}
\label{subsec: upperlower}

\begin{Thm}
For $q\geq \log_d \left(\frac{1+\sqrt{5}}{2}\right)+1$ and any two-qudit state $\rho_{AB}$, we have
\begin{align}
{\mathbf u}E_q^{\leftarrow}(\rho_{AB})\leq \frac{{\mathcal I}_q\left(\rho_{AB}\right)}{2}
\label{uEup}
\end{align}
and
\begin{align}
{\mathbf u}\delta_q^{\leftarrow}(\rho_{AB})\geq \frac{{\mathcal I}_q\left(\rho_{AB}\right)}{2}
\label{uDlow}
\end{align}
conditioned on the subadditivity of Tsallis-$q$ mutual entropy for the ccq states in Eq.~(\ref{XYAB}), that is,
\begin{align}
{\mathcal I}_q\left(\Omega_{XY:AB}\right)\geq {\mathcal I}_q\left(\Omega_{X:AB}\right)+{\mathcal I}_q\left(\Omega_{Y:AB}\right).
\label{subadd1}
\end{align}
\label{thm: uEup}
\end{Thm}

\begin{proof}
For a two-qudit state $\rho_{AB}$ and its four-party ccq state defined in Eq.~(\ref{XYAB}),
Eqs.~(\ref{IqXYAB2}), (\ref{IqXAB2}) and (\ref{IqYAB2}) enable us to rewrite Inequality~(\ref{subadd1}) as
\begin{align}
\chi_q(\mathcal E_0)+\chi_q(\mathcal E_1)\leq& S_q\left(\rho_A\right)+S_q\left(\rho_B\right)\nonumber\\
&-d^{1-q}S_q\left(\rho_{AB}\right)+\frac{\left(d^{1-q}-1\right)^2}{d^{1-q}(1-q)}.
\label{ensemine1}
\end{align}

Now we note that the Tsallis-$q$ differences $\chi_q(\mathcal E_0)$ and $\chi_q(\mathcal E_1)$ for the ensembles $\mathcal E_0$ and $\mathcal E_1$ in Eq.~(\ref{ensembles}) can be obtained from $\rho_{AB}$ by measuring its subsystem $B$ with respect to
the rank-1 measurements $\{\ket{e_i}_B\bra{e_i} \}_i$
and $\{ |\tilde e_j \rangle_B \langle \tilde e_j| \}_j$, respectively.
From the definition of $q$-UE in Eq.~(\ref{qUE}), we also note that each of rank-1 measurement $\{\ket{e_i}_B\bra{e_i} \}_i$
and $\{ |\tilde e_j \rangle_B \langle \tilde e_j| \}_j$ provides an upperbound of $q$-UE as
\begin{align}
{\mathbf u}E_q^{\leftarrow}(\rho_{AB})\leq \chi_q(\mathcal E_0)
\label{tuupper1a}
\end{align}
and
\begin{align}
{\mathbf u}E_q^{\leftarrow}(\rho_{AB})\leq \chi_q(\mathcal E_1).
\label{tuupper1b}
\end{align}

Because Inequalities~(\ref{tuupper1a}) and (\ref{tuupper1b}) imply
\begin{align}
{\mathbf u}E_q^{\leftarrow}(\rho_{AB})\leq \frac{\chi_q(\mathcal E_0)+\chi_q(\mathcal E_1)}{2},
\label{tuupper1}
\end{align}
Inequalities~(\ref{ensemine1}) and (\ref{tuupper1}) enable us to have
\begin{align}
{\mathbf u}E_q^{\leftarrow}(\rho_{AB})\leq& \frac{{\mathcal I}_q\left(\rho_{AB}\right)}{2}\nonumber\\
&+\frac{1}{2}\left[(1-d^{1-q})S_q\left(\rho_{AB}\right)+
\frac{\left(d^{1-q}-1\right)^2}{d^{1-q}(1-q)}\right].
\label{upperUEAB}
\end{align}
To prove Inequality~(\ref{uEup}), it is now sufficient to show that
\begin{align}
(1-d^{1-q})S_q\left(\rho_{AB}\right)+\frac{\left(d^{1-q}-1\right)^2}{d^{1-q}(1-q)}\leq 0.
\label{nega1}
\end{align}

For $q\geq 1$, we have $1-d^{1-q}\geq0$. Moreover, the Tsallis $q$-entropy attains its maximal value for the maximally mixed state,
\begin{align}
S_q\left(\rho_{AB}\right)\leq S_q\left(\frac{I_{AB}}{d^2}\right)=\frac{1-d^{2(1-q)}}{q-1},
\label{Tsalmax}
\end{align}
therefore we have
\begin{align}
(1-d^{1-q})&S_q\left(\rho_{AB}\right)+\frac{\left(d^{1-q}-1\right)^2}{d^{1-q}(1-q)}\nonumber\\
&\leq \frac{(1-d^{1-q})^2}{q-1}\left[1+d^{1-q}-d^{q-1}\right].
\label{nega2}
\end{align}

The non-positivity of the right-hand side of Inequality~(\ref{nega2}) is equivalent to
\begin{align}
1+d^{1-q}-d^{q-1}\leq 0,
\label{nonposi1}
\end{align}
which can be rewritten as
\begin{align}
q\geq \log_d \left(\frac{1+\sqrt{5}}{2}\right)+1
\label{nonposi2}
\end{align}
for nonnegative $q$.

Inequality~(\ref{uDlow}) is then a one step consequence of Inequality~(\ref{uEup}) together with the definition of $q$-UD in Eq.~(\ref{qudis}).
\end{proof}

For $q=1$, Inequalities~(\ref{uEup}) and (\ref{uDlow}) are reduced to
\begin{align}
{\mathbf u}E^{\leftarrow}(\rho_{AB})&\leq \frac{{\mathcal I}\left(\rho_{AB}\right)}{2},~~
{\mathbf u}\delta^{\leftarrow}(\rho_{AB})\geq \frac{{\mathcal I}\left(\rho_{AB}\right)}{2},
\label{Eup}
\end{align}
respectively, whereas the condition in~(\ref{subadd1}) is reduced to the subadditivity of quantum mutual information
\begin{align}
{\mathcal I}\left(\Omega_{XY:AB}\right)\geq {\mathcal I}\left(\Omega_{X:AB}\right)+{\mathcal I}\left(\Omega_{Y:AB}\right).
\label{mutsubadd}
\end{align}
In fact, Inequality~(\ref{mutsubadd}) was shown to be true for any ccq state in general~\cite{Kim16T}. Moreover, Inequalities~(\ref{Eup}) are also shown to be true for any quantum state $\rho_{AB}$~\cite{BGK, XFL12}.
Thus Theorem~\ref{thm: uEup} is true for $q=1$ without any condition.

The lower bound(the right-hand side) of Inequality~(\ref{nonposi2}) tends to $1$ as $d$ is getting large.
Thus Theorem~\ref{thm: uEup} is true for the most range of $q\geq 1$ if $d$ is large enough, that is, large dimensional quantum systems.
Although the proof method that we used here is not sufficient to guarantee the validity of
Theorem~\ref{thm: uEup} for $1<q<\log_d \left(\frac{1+\sqrt{5}}{2}\right)$, we conjecture that Theorem~\ref{thm: uEup} is true for any $q$ larger than or equal to $1$.

We also note that any bipartite quantum state can be considered as a two-qudit state where $d$ is the dimension of larger dimensional subsystem.
Moreover, we also have $\log_d \left(\frac{1+\sqrt{5}}{2}\right)\leq1$ for any $ d \geq 2$, Thus we have the following corollary.
\begin{Cor}
For $q\geq 2$ and any bipartite quantum state $\sigma_{AB}$, we have
\begin{align}
{\mathbf u}E_q^{\leftarrow}(\sigma_{AB})\leq \frac{{\mathcal I}_q\left(\sigma_{AB}\right)}{2}
\label{uEup2}
\end{align}
and
\begin{align}
{\mathbf u}\delta_q^{\leftarrow}(\sigma_{AB})\geq \frac{{\mathcal I}_q\left(\sigma_{AB}\right)}{2}
\label{uDlow2}
\end{align}
conditioned on the subadditivity of Tsallis-$q$ mutual entropy for the ccq state in terms of $\sigma_{AB}$.
\label{Cor: uEup}
\end{Cor}

\section{Polygamy of $q$-expected entanglement in multi-party quantum systems}
\label{sec: poly}
In this section, we provide the polygamy inequalities of $q$-expected quantum entanglement distributed in
multi-party quantum systems for $q\geq 1$ conditioned on the subadditivity of Tsallis-$q$ mutual entropy for
ccq states. The following theorem shows the polygamy inequality of $q$-EOA in three-party quantum systems.

\begin{Thm}
For $q\geq 1$, and any three-party pure state $\ket{\psi}_{ABC}$
with its two-party reduced density matrices $\T_C \ket{\psi}_{ABC}\bra{\psi}_{ABC}=\rho_{AB}$ and $\T_B \ket{\psi}_{ABC}\bra{\psi}_{ABC}=\rho_{AC}$,
we have
\begin{align}
E_{q}\left(\ket{\psi}_{A(BC)}\right)
\leq& E^a_{q}\left(\rho_{AB}\right)+E^a_{q}\left(\rho_{AC}\right),
\label{Tqpoly3}
\end{align}
conditioned on the subadditivity of Tsallis-$q$ mutual entropy for the ccq states in Eq.~(\ref{XYAB}),
where $E_{q}\left(\ket{\psi}_{A(BC)}\right)$ is the $q$-E of the pure state $\ket{\psi}_{ABC}$ with respect to the bipartition between $A$ and $BC$
\label{thm: Tqpoly3}
\end{Thm}

\begin{proof}
For a three-party quantum state $\ket{\psi}_{ABC}$, the universality of Eq.~(\ref{qUEEqa}) of Proposition~\ref{thm: qUEEqa}
leads us to
\begin{align}
S_q(\rho_A)-{\mathbf u}E_q^{\leftarrow}(\rho_{AC})=E^a_q\left(\rho_{AB}\right)
\end{align}
and
\begin{align}
S_q(\rho_A)-{\mathbf u}E_q^{\leftarrow}(\rho_{AB})=E^a_q\left(\rho_{AC}\right),
\label{qUEEqa23}
\end{align}
therefore we have
\begin{align}
2S_q(\rho_A)-&({\mathbf u}E_q^{\leftarrow}(\rho_{AB})+{\mathbf u}E_q^{\leftarrow}(\rho_{AC}))\nonumber\\
&=E^a_q\left(\rho_{AB}\right)+E^a_q\left(\rho_{AC}\right).
\label{qEqUEineq1}
\end{align}
We also note that $\ket{\psi}_{ABC}$ can be assumed to be a three-qudit state,
otherwise, we can always consider an imbedded image of $\ket{\psi}_{ABC}$ into a higher dimensional quantum system having the same dimensions of subsystems.

As we have already seen in the proof of Theorem~\ref{thm: uEup}, the subadditivity condition for ccq states leads us to Inequality~(\ref{ensemine1}),
and this enables us to have an upper bound of ${\mathbf u}E_q^{\leftarrow}(\rho_{AB})$ as in Inequality~(\ref{upperUEAB}), which can be rewritten as
\begin{align}
{\mathbf u}E_q^{\leftarrow}(\rho_{AB})\leq& \frac{1}{2}[S_q\left(\rho_A\right)+S_q\left(\rho_B\right)\nonumber\\
&-d^{1-q}S_q\left(\rho_{AB}\right)+\frac{\left(d^{1-q}-1\right)^2}{d^{1-q}(1-q)}].
\label{upperUEABsim}
\end{align}
For the two-qudit reduced density matrix $\rho_{AC}$, we also analogously have
\begin{align}
{\mathbf u}E_q^{\leftarrow}(\rho_{AC})\leq& \frac{1}{2}[S_q\left(\rho_A\right)+S_q\left(\rho_C\right)\nonumber\\
&-d^{1-q}S_q\left(\rho_{AC}\right)+\frac{\left(d^{1-q}-1\right)^2}{d^{1-q}(1-q)}].
\label{upperUEACsim}
\end{align}

From Inequality~(\ref{qEqUEineq1}) together with inequalities~(\ref{upperUEABsim}) and (\ref{upperUEACsim}), we have
\begin{align}
S_q(\rho_A)
+\frac{1}{2}\left(\Xi_B+\Xi_C\right) \leq E^a_q\left(\rho_{AB}\right)+E^a_q\left(\rho_{AC}\right)
\label{qEqUEineq2}
\end{align}
where
\begin{align}
\Xi_B=\frac{d^{q-1}-1}{d^{q-1}}\left[\frac{d^{q-1}-1}{q-1}-S_q\left(\rho_B \right)\right]
\label{XiB}
\end{align}
and
\begin{align}
\Xi_C=\frac{d^{q-1}-1}{d^{q-1}}\left[\frac{d^{q-1}-1}{q-1}-S_q\left(\rho_C \right)\right].
\label{XiC}
\end{align}

For $q\geq 1$, we have $\frac{d^{q-1}-1}{d^{q-1}}\geq0$. Moreover,the Tsallis $q$-entropy attains its maximal value for the maximally mixed states,
\begin{align}
S_q\left(\rho_{B}\right)\leq S_q\left(\frac{I_{B}}{d}\right)=\frac{1-d^{(1-q)}}{q-1}\leq \frac{d^{(q-1)}-1}{q-1},
\label{Tsalmax2}
\end{align}
and this implies the nonnegativity of $\Xi_B$. The nonnegativity of $\Xi_C$ can be analogously obtained, therefore
\begin{align}
\Xi_B\geq0,~~\Xi_C\geq0.
\label{noneg12}
\end{align}
Because $E_{q}\left(\ket{\psi}_{A(BC)}\right)=S_q(\rho_A)$, Inequality~(\ref{qEqUEineq2}) together with the inequalities in (\ref{noneg12}) implies Inequality~(\ref{Tqpoly3}), which completes the proof.
\end{proof}

When $q$ tends to $1$, $q$-EOA is reduced to EOA as in Eq.~(\ref{TsallistEOA}), whereas the subadditivity of quantum mutual information in
Inequality~(\ref{mutsubadd}) was shown to be true for any ccq state in general~\cite{Kim16T}.
Thus Theorem~\ref{thm: Tqpoly3} encapsulates the results of general polygamy inequality of three-party entanglement in terms of EOA~\cite{BGK}.

For any three-party quantum state $\ket{\psi}_{ABC}$, it was recently shown that the polygamy inequality of $q$-EOA in (\ref{Tqpoly3}) is a necessary and sufficient condition for the monogamy inequality of $q$-UE as well as the polygamy inequality of $q$-UD for $q\geq 1$~\cite{Kim19}. Thus we have the following corollary.
\begin{Cor}
For $q\geq 1$, and any three-party pure state $\ket{\psi}_{ABC}$, we have
\begin{align}
{\mathbf u}E_q^{\leftarrow}(\ket{\psi}_{A(BC)})\geq {\mathbf u}E_q^{\leftarrow}(\rho_{AB})+{\mathbf u}E_q^{\leftarrow}(\rho_{AC}),
\label{UEmono3}
\end{align}
and
\begin{align}
{\mathbf u}\delta_q^{\leftarrow}\left(\ket{\psi}_{A(BC)}\right)
\leq& {\mathbf u}\delta_q^{\leftarrow}(\rho_{AB})+{\mathbf u}\delta_q^{\leftarrow}(\rho_{AC}),
\label{Dqpoly3}
\end{align}
conditioned on the subadditivity of Tsallis-$q$ mutual entropy for the ccq state in Eq.~(\ref{XYAB}).
\label{Cor: qUEmono3}
\end{Cor}

Now, we generalize Theorem~\ref{thm: Tqpoly3} for an arbitrary multi-party quantum system.
\begin{Thm}
For $q\geq 1$, and any multi-party quantum state $\rho_{A_1A_2\cdots A_n}$ with two-party reduced density matrices
$\rho_{A_1A_i}$ for $i=2, \cdots , n$, we have
\begin{align}
E^a_{q}\left(\rho_{A_1(A_2\cdots A_n)}\right)
\leq& \sum_{i=2}^{n}E^a_{q}\left(\rho_{A_1A_i}\right),
\label{Tqpolyn}
\end{align}
conditioned on the subadditivity of Tsallis-$q$ mutual entropy for the ccq states in Eq.~(\ref{XYAB}).
\label{thm: Tqpolyn}
\end{Thm}

\begin{proof}
We first prove the theorem for any three-party mixed state $\rho_{ABC}$, then the validity of the theorem for an arbitrary $n$-party quantum state $\rho_{A_1A_2\cdots A_n}$ follows inductively.

For a three-party mixed state $\rho_{ABC}$, let us consider an optimal decomposition of
$\rho_{ABC}$ for $q$-EOA with respect to the bipartition between $A$ and $BC$, that is,
\begin{align}
\rho_{ABC}=\sum_i p_i\ket{\psi_i}_{ABC}\bra{\psi_i},
\label{opt1}
\end{align}
with
\begin{align}
E^a_{q}\left(\rho_{A(BC)}\right)=\sum_i p_i^q E_{q}\left(\ket{\psi_i}_{A(BC)}\right).
\label{optTEOA1}
\end{align}
From Theorem~\ref{thm: Tqpoly3}, each $\ket{\psi_i}_{ABC}$ in Eq.~(\ref{optTEOA1}) satisfies
\begin{align}
E_{q}\left(\ket{\psi_i}_{A(BC)}\right)\leq E^a_{q}\left(\rho^i_{AB}\right)+E^a_{q}\left(\rho^i_{AC}\right)
\label{polyi}
\end{align}
with $\rho^i_{AB}=\T_C \ket{\psi_i}_{ABC}\bra{\psi_i}$ and $\rho^i_{AC}=\T_B \ket{\psi_i}_{ABC}\bra{\psi_i}$,

For each $i$ and the two-qudit reduced density matrices $\rho_{AB}^i$ and $\rho_{AC}^i$,
let us consider their optimal decompositions for $q$-EOA, that is,
\begin{align}
\rho_{AB}^i=&\sum_{j}r_{ij}\ket{\phi_j^i}_{AB}\bra{\phi_j^i}~,
\rho_{AC}^i=&\sum_{l}s_{il}\ket{\mu_l^i}_{AB}\bra{\mu_l^i},
\label{optrho^i}
\end{align}
such that
\begin{align}
E_q^a\left(\rho_{AB}^i\right)=&\sum_{j}r_{ij}^q E_q\left(\ket{\phi_j^i}_{AB}\right)\nonumber\\
E_q^a\left(\rho_{AC}^i \right)=&\sum_{l}s_{il}^q E_q\left(\ket{\mu_l^i}_{AB}\right).
\label{qErho^i}
\end{align}

Now, we have
\begin{widetext}
\begin{align}
E_{q}^a\left(\rho_{A(BC)}\right)=&\sum_i p_i^q E_{q}\left(\ket{\psi_i}_{A(BC)}\right)\nonumber\\
\leq&\sum_i p_i^q\left( E^a_{q}\left(\rho^i_{AB}\right)+E^a_{q}\left(\rho^i_{AC}\right)\right)\nonumber\\
=&\sum_i p_i^q\left(\sum_{j}r_{ij}^q E_q\left(\ket{\phi_j^i}_{AB}\right)+
\sum_{l}s_{il}^q E_q\left(\ket{\mu_l^i}_{AB}\right)\right)\nonumber\\
=&\sum_{i,j} \left(p_i r_{ij}\right)^q E_q\left(\ket{\phi_j^i}_{AB}\right)+\sum_{i,l}\left(p_i s_{il}\right)^q E_q\left( \ket{\mu_l^i}_{AB}\right)\nonumber\\
\leq& E_{q}^a\left(\rho_{AB}\right)+E_{q}^a\left(\rho_{AC}\right),
\label{Tqpoly3mix}
\end{align}
\end{widetext}
where the first inequality is from Inequality~(\ref{polyi}), and the second inequality is due to
\begin{align}
\rho_{AB}=\sum_i p_i \rho_{AB}^i=&\sum_{i,j}p_i r_{ij}\ket{\phi_j^i}_{AB}\bra{\phi_j^i},\nonumber\\
\rho_{AC}=\sum_i p_i \rho_{AC}^i=&\sum_{i,l}p_i s_{il}\ket{\mu_l^i}_{AB}\bra{\mu_l^i}
\label{decompall}
\end{align}
and the definition of $q$-EOA in Eq.~(\ref{qEOA}). Thus Inequality~(\ref{Tqpolyn}) is true for three-party mixed states.

For general multi-party quantum system, we use the mathematical induction on the number of parties $n$; let us assume the polygamy inequality~(\ref{Tqpolyn}) is true for any $(n-1)$-party quantum state, and consider an $n$-party quantum state $\rho_{A_1A_2\cdots A_n}$ for $n \geq 4$. By considering $\rho_{A_1A_2\cdots A_n}$ as a three-party state with respect to the partition $A_1$, $A_2$ and $A_3\cdots A_n$,
Inequality~(\ref{Tqpoly3mix}) leads us to
\begin{align}
E^a_{q}\left(\rho_{A_1(A_2\cdots A_n)}\right)
\leq& E^a_{q}\left(\rho_{A_1A_2}\right)+E^a_{q}\left(\rho_{A_1(A_3\cdots A_n)}\right)
\label{polymixed1}
\end{align}
where $\rho_{A_1A_3\cdots A_n}=\T_{A_2}\rho_{A_1A_2\cdots A_n}$.

Because $\rho_{A_1A_3\cdots A_n}$ in Inequality~(\ref{polymixed1}) is a
$(n-1)$-party quantum state, the induction hypothesis assures that
\begin{align}
E^a_{q}\left(\rho_{A_1(A_3\cdots A_n)}\right) \leq
E^a_{q}\left(\rho_{A_1A_3}\right)+\cdots +E^a_{q}\left(\rho_{A_1A_n}\right).
\label{n-1poly}
\end{align}
Thus Inequalities~(\ref{polymixed1}) and (\ref{n-1poly}) lead us to the polygamy
inequality of multi-party entanglement in terms of $q$-EOA in (\ref{Tqpolyn}).
\end{proof}

As $q$-EOA is reduced to EOA  when $q$ tends to $1$, and due to the
the subadditivity of quantum mutual information for CCQ states,
Theorem~\ref{thm: Tqpolyn} encapsulates the results of general polygamy inequality of multi-party entanglement in terms of EOA~\cite{KimGP}.

%%%%%%%%%%%%%%%%%%%%%%%%%%%%%%%%%%%%%%%%%%%%%%%%%%%%%%%%%%%%%%%%%%%%%%
%%%                                                                %%%
%%%                           Conclusion                           %%%
%%%                                                                %%%
%%%%%%%%%%%%%%%%%%%%%%%%%%%%%%%%%%%%%%%%%%%%%%%%%%%%%%%%%%%%%%%%%%%%%%

\section{Conclusion}
\label{Conclusion}
We have characterized the polygamy property of multi-party quantum entanglement in terms of $q$-expectation value for the full range of $q\geq 1$. By using the generalized definitions of various classical and quantum correlations in terms of Tsallis $q$-entropy and $q$-expectation value, we have first provided some properties of $q$-expected correlations in relation to classical-classical-quantum(ccq) states. Based on these properties, we have established a class of polygamy inequalities of multi-party quantum entanglement in arbitrary dimensions in terms of $q$-EOA.

Due to the equivalence between the monogamy of $q$-UE and polygamy of $q$-EOA and $q$-UD, our results also guarantee the monogamy inequality of $q$-UE as well as the polygamy inequality of $q$-UD distributed in three-party quantum systems. We also note that, from the continuity of $q$-expectation value as well as Tsallis $q$-entropy, our results encapsulate the previous results of monogamy and polygamy inequalities based on von Neumann entropy as special cases.

Studying multi-party quantum correlations, especially in higher dimensions more than qubits,
is important and necessary for various reasons. In many quantum information processing tasks such as quantum communication and quantum cryptography, higher-dimensional quantum systems are sometimes considered to be more useful because they can provide higher coding density and thus stronger security compared with qubit systems. We also note that the monogamy and polygamy properties of multi-party entanglement play
a central role in quantum cryptography because they can bound the possible amount of correlation between the authenticated users and the eavesdropper; the fundamental concept of the security proof. Thus our results about monogamy and polygamy inequalities of $q$-expected quantum correlations in arbitrary high-dimensional quantum systems can provide good methods and rich references for the foundation of many secure quantum information processing tasks.

\section*{Acknowledgments}
This research was supported by Basic Science Research Program through the National Research Foundation of Korea(NRF) funded by the Ministry of Education(NRF-2017R1D1A1B03034727).

%%%%%%%%%%%%%%%%%%%%%%%%%%%%%%%%%%%%%%%%%%%%%%%%%%%%%%%%%%%%%%%%%%%%%%%%

\end{document}